\documentclass[11pt]{article}

\usepackage[margin=1.0in]{geometry}
\usepackage{amssymb}
\usepackage{amsmath}
\usepackage{algpseudocode}
\usepackage{algorithm}
\usepackage{graphicx}
\newtheorem{theorem}{Theorem}

\newtheorem{problem}{Problem}

\newenvironment{proof}[1][Proof]{\begin{trivlist}
\item[\hskip \labelsep {\bfseries #1}]}{\end{trivlist}}
\newenvironment{proofsketch}[1][Sketch of proof]{\begin{trivlist}
\item[\hskip \labelsep {\bfseries #1}]}{\end{trivlist}}
\newcommand{\qed}{\nobreak \ifvmode \relax \else
      \ifdim\lastskip<1.5em \hskip-\lastskip
      \hskip1.5em plus0em minus0.5em \fi \nobreak
      \vrule height0.75em width0.5em depth0.25em\fi}

\begin{document}

\title{\bf Exact Learning of RNA Energy Parameters From Structure}
\author{Hamidreza Chitsaz\footnote{Corresponding author; e-mail: chitsaz@wayne.edu}\\
Department of Computer Science\\
Wayne State University\\
Detroit, MI
\and Mohammad Aminisharifabad\\
Department of Computer Science\\
Wayne State University\\
Detroit, MI}

\maketitle

\begin{abstract}
We consider the problem of exact learning of parameters of a linear RNA energy model from secondary structure data. A necessary and sufficient condition for learnability of parameters is derived, which is based on computing the convex hull of union of translated Newton polytopes of input sequences \cite{Forouzmand13}. The set of learned energy parameters is characterized as the convex cone generated by the normal vectors to those facets of the resulting polytope that are incident to the origin. In practice, the sufficient condition may not be satisfied by the entire training data set; hence, computing a maximal subset of training data for which the sufficient condition is satisfied is often desired. We show that problem is NP-hard in general for an arbitrary dimensional feature space. Using a randomized greedy algorithm, we select a subset of RNA STRAND v2.0 database that satisfies the sufficient condition for separate A-U, C-G, G-U base pair counting model. The set of learned energy parameters includes 
experimentally measured energies of A-U, C-G, and G-U pairs; hence, our parameter set is in agreement with the Turner parameters.
\end{abstract}

\section{Introduction}
The discovery of key regulatory roles of RNA in the cell has recently invigorated interest in RNA structure and RNA-RNA interaction determination or prediction \cite{Storz02,Bartel04,Hannon02,ZamHal05,WagFla02,Bra02,Gottesman05}. Due to high chemical reactivity of nucleic acids, experimental determination of RNA structure is time-consuming and challenging. In spite of the fact that computational prediction of RNA structure may not be accurate in a significant number of cases, it is the only viable low-cost high-throughput option to date. Furthermore, with the advent of whole genome synthetic biology \cite{Gibson10}, accurate high-throughput RNA engineering algorithms are required for both \emph{in vivo} and \emph{in vitro} applications \cite{See05,SeeLuk05,SimDit05,Ven07,Reif08}. 

Since the dawn of RNA secondary structure prediction four decades ago \cite{Tinoco73}, increasingly complex models and algorithms have been proposed. Early approaches considered mere base pair counting \cite{Nus78,WatSmi78}, followed by the Turner thermodynamics model which was a significant leap forward \cite{Zuker81,Mcc90}. Recently, massively feature-rich models empowered by parameter estimation 
algorithms have been proposed. Despite significant progress in the last three decades, made possible by the work of Turner and others \cite{MatTur99} on measuring RNA thermodynamic energy parameters and the work of several groups on novel algorithms \cite{RivEdd99,DirPie03,Chitsaz09,Chitsaz09b,BerTaf06,Siederdissen11,Huang10,Backofen11} and machine learning approaches \cite{Do06,Andronescu10,Zakov11}, the RNA structure prediction accuracy has not reached a satisfactory level yet \cite{Rivas12}. 

Until now, computational convenience, namely the ability to develop dynamic programming algorithms of polynomial running time, and biophysical intuition have played the main role in development of RNA energy models. 
The first step towards high accuracy RNA structure prediction is to make sure that the energy model is inherently capable of predicting every observed structure, but it is not over-capable, as accurate estimation of the parameters of a highly complex model often requires a myriad of experimental data, and lack of sufficient experimental data causes overfitting. Recently, we gave a systematic method to assess that inherent capability of a given energy model \cite{Forouzmand13}. Our algorithm decides whether the parameters of an energy model are \emph{learnable}. The parameters of an energy model are defined to be \emph{learnable} iff there exists at least one set of such parameters that renders \emph{every} known RNA structure to date, determined through X-ray or NMR, the minimum free energy structure. Equivalently, we say that the parameters of an energy model are learnable iff 100\% structure prediction accuracy can be achieved when the training and test sets are identical. Previously, we gave a necessary 
condition for the learnability and an algorithm to verify it. Note that a successful RNA folding algorithm needs to have generalization power to predict unseen structures. We leave assessment of the generalization power for future work. 

In this paper, we give a necessary and sufficient condition for the learnability and characterize the set of learned energy parameters. Also, we show that selecting the maximum number of RNA sequences for which the sufficient condition is satisfied is NP-hard in general in arbitrary dimensions. Using a randomized greedy algorithm, we select a subset of RNA STRAND v2.0 database that satisfies the sufficient condition for separate A-U, C-G, G-U base pair counting model and yields a set of learned energy parameters that includes the experimentally measured energies of A-U, C-G, and G-U pairs.

\section{Methods}

\subsection{Preliminaries}
Let the training set $D = \left\{(x_i, y_i)\ |\ i=1,2,\ldots,n \right\}$ be a given collection of RNA sequences $x$ and their corresponding experimentally observed structures $y$. Throughout this paper, we assume the free energy 
\begin{equation}\label{equ:energy}
G(x, s, \mathbf{h}) := \left\langle c(x, s), \mathbf{h} \right\rangle 
\end{equation}
associated with a sequence-structure pair $(x, s)$ is a linear function of the energy parameters $\mathbf{h} \in \mathbb{R}^k$, in which $k$ is the number of features, $s$ is a structure in $\mathcal{E}(x)$ the ensemble of possible structures of $x$, and $c(x, s) \in \mathbb{Z}^k$ is the feature vector. 

\subsection{Learnability of energy parameters}
The question that we asked before \cite{Forouzmand13} was: does there exist nonzero parameters $\mathbf{h}^\dagger$ such that for every $(x, y) \in D$, $y = \arg\min_s G(x, s, \mathbf{h}^\dagger)$? We ask a slightly relaxed version of that question in this paper: does there exist nonzero parameters $\mathbf{h}^\dagger$ such that for every $(x, y) \in D$, $G(x, y, \mathbf{h}^\dagger) = \min_s G(x, s, \mathbf{h}^\dagger)$? The answer to this question reveals inherent limitations of the energy model, which can be used to design improved models. We provided a necessary condition for the existence of $\mathbf{h}^\dagger$ and a dynamic programming algorithm to verify it through computing the Newton polytope for every $x$ in $D$ \cite{Forouzmand13}. Following our previous notation, let the \emph{feature ensemble} of sequence $x$ be
\begin{equation}
\mathcal{F}(x) := \left\{ c(x, s)\ |\ s \in \mathcal{E}(x) \right\} \subset \mathbb{Z}^k,
\end{equation}
and call the convex hull of $\mathcal{F}(x)$,
\begin{equation}
\mathcal{N}(x) := \mbox{conv}\mathcal{F}(x) \subset \mathbb{R}^k,
\end{equation}
the \emph{Newton polytope} of $x$. We remind the reader that the convex hull of a set, denoted by `conv' here, is the minimal convex set that fully contains the set. Let $(x, y) \in D$ and $0 \not= \mathbf{h}^\dagger \in \mathbb{R}^k$. We previously showed in \cite{Forouzmand13} that if $y$ minimizes $G(x, s, \mathbf{h}^\dagger)$ as a function of $s$, then $c(x, y) \in \partial \mathcal{N}(x)$, i.e. the feature vector of $(x, y)$ is on the boundary of the Newton polytope of $x$. 

In this paper, we provide a necessary and sufficient condition for the existence of $\mathbf{h}^\dagger$. First, we rewrite that necessary condition by introducing a translated copy of the Newton polytope,
\begin{equation}
\mathcal{N}_y(x) := \mathcal{N}(x)\ominus c(x, y) = \mbox{conv}\left\{\mathcal{F}(x)\ominus c(x, y)\right\},
\end{equation}
in which $\ominus$ is the Minkowski difference. The necessary condition for learnability then becomes $0 \in \partial \mathcal{N}_y(x)$.

\subsection{Necessary and sufficient condition for learnability}
The following theorem specifies a necessary and sufficient condition for the learnability.
\begin{theorem}[Necessary and Sufficient Condition]\label{thm:cond}
There exists $0 \not= \mathbf{h}^\dagger \in \mathbb{R}^k$ such that for all $(x, y) \in D$, $y$ minimizes $G(x, s, \mathbf{h}^\dagger)$ as a function of $s$ iff $0 \in \partial \mathcal{N}(D)$, in which
\begin{equation} 
\mathcal{N}(D) := \mbox{\emph{conv}} \left\{ \bigcup_{(x, y) \in D} \mathcal{N}_y(x) \right\} = \mbox{\emph{conv}} \left\{ \bigcup_{(x, y) \in D} \mathcal{F}(x)\ominus c(x, y) \right\}.
\end{equation} 
\end{theorem}
\begin{proof}
($\Rightarrow$) Suppose $0 \not= \mathbf{h}^\dagger \in \mathbb{R}^k$ exists such that for all $(x, y) \in D$, $y$ minimizes $G(x, s, \mathbf{h}^\dagger)$ as a function of $s$. To the contrary, assume $0$ is in the interior of $\mathcal{N}(D)$. Therefore, there is an open ball of radius $\delta > 0$ centered at $0 \in \mathbb{Z}^k$ completely contained in $\mathcal{N}(D)$, i.e.
\begin{equation}
B_\delta(0) \subset \mathcal{N}(D).
\end{equation}
Let $$p = - (\delta/2)\frac{\mathbf{h}^\dagger}{\|\mathbf{h}^\dagger\|}.$$ It is clear that $p \in B_\delta(0) \subset \mathcal{N}(D)$ since $\|p\| = \delta/2 < \delta$. Therefore, $p$ can be written as a convex linear combination of the feature vectors in 
\begin{equation}\label{equ:totalensemble}
\mathcal{F}(D) := \bigcup_{(x, y) \in D} \left\{\mathcal{F}(x)\ominus c(x, y) \right\} = \left\{v_1, \ldots, v_N \right\}, 
\end{equation}
i.e.
\begin{align}
\exists\; \ \alpha_1, \ldots \alpha_N \geq 0:\; \ \alpha_1 v_1 + \cdots + \alpha_N v_N &= p,\\
\alpha_1 + \cdots + \alpha_N &= 1. 
\end{align}
Note that
\begin{equation}\label{equ:negative}
\left\langle p, \mathbf{h}^\dagger\right\rangle = -(\delta/2) \|\mathbf{h}^\dagger\| < 0.
\end{equation}
Therefore, there is $1 \leq i \leq N$, such that $\left\langle v_i, \mathbf{h}^\dagger\right\rangle < 0$ for otherwise,
\begin{equation}
\left\langle p, \mathbf{h}^\dagger\right\rangle = \sum_{i=1}^N \alpha_i \left\langle v_i, \mathbf{h}^\dagger\right\rangle \geq 0,
\end{equation}
which would be a contradiction with (\ref{equ:negative}). Since $v_i \in \mathcal{F}(D)$ in (\ref{equ:totalensemble}), there is $(x, y) \in D$ such that $v_i \in \mathcal{F}(x)\ominus c(x, y)$. It is now sufficient to note that $v'_i = v_i + c(x, y) \in \mathcal{F}(x)$ and $\left\langle v'_i, \mathbf{h}^\dagger\right\rangle < \left\langle c(x, y), \mathbf{h}^\dagger\right\rangle = G(x, y, \mathbf{h}^\dagger)$ which is a contradiction with our assumption that $y$ minimizes $G(x, s, \mathbf{h}^\dagger)$ as a function of $s$. \\

\noindent ($\Leftarrow$) Suppose $0$ is on the boundary of $\mathcal{N}(D)$. Note that for all $(x, y) \in D$, $0 \in \partial \mathcal{N}_y(x)$.  We construct a nonzero $\mathbf{h}^\dagger \in \mathbb{R}^k$ such that for all $(x, y) \in D$, $G(x, y, \mathbf{h}^\dagger) = \min_{s\in \mathcal{E}(x)} G(x, s, \mathbf{h}^\dagger)$. Since $\mathcal{N}(D)$ is convex, it has a supporting hyperplane $\mathcal{H}$, which passes through $0$. Let the positive normal to $\mathcal{H}$ be $\mathbf{h}^\dagger$, i.e. $\langle \mathcal{N}(D), \mathbf{h}^\dagger\rangle \subset [0, \infty)$. Therefore, $\min_{p \in \mathcal{N}(D)} \langle p, \mathbf{h}^\dagger\rangle = \langle 0, \mathbf{h}^\dagger\rangle = 0,$
which implies that for all $(x, y) \in D$,
$\min_{p \in \mathcal{N}_y(x)} \langle p, \mathbf{h}^\dagger\rangle = \langle 0, \mathbf{h}^\dagger\rangle = 0,$
or equivalently, $G(x, y, \mathbf{h}^\dagger) = \min_{s \in \mathcal{E}(x)} G(x, s, \mathbf{h}^\dagger)$. \qed
\end{proof}
The proof above is constructive; hence using a similar argument, we characterize all learned energy parameters in the following theorem. 
\begin{theorem}\label{thm:cone}
Let
\begin{equation}
H(D) := \left\{ \mathbf{h}^\dagger \in \mathbb{R}^k \ | \ \mathbf{h}^\dagger \not = 0, G(x, y, \mathbf{h}^\dagger) = \min_{s \in \mathcal{E}(x)} G(x, s, \mathbf{h}^\dagger)\;\ \ \forall (x, y) \in D \right\}.
\end{equation}
In that case, $H(D)$ is the set of vectors $n \in \mathbb{R}^k$ orthogonal to the supporting hyperplanes of $\mathcal{N}(D)$ such that $\langle \mathcal{N}(D), n\rangle \subset [0, \infty)$. Moreover, that is the convex cone generated by the inward normal vectors to those facets of $\mathcal{N}(D)$ that are incident to $0$.
\end{theorem}

\subsection{Compatible training set}
We say that a training set $D$ is \emph{compatible} if the sufficient condition for learnability is satisfied when $D$ is considered. However, the sufficient condition for the entire training set is often not satisfied in practice, for example when the feature vector is low dimensional. We would like to find a compatible subset of $D$ to estimate the energy parameters. A natural quest is to find the maximal compatible subset of $D$. In the following section, we show that even when the Newton polytopes with polynomial complexity for all of the sequences in $D$ are given, selection of a maximal compatible subset of $D$ is NP-hard in arbitrary dimensions.

\subsection{NP-hardness of maximal compatible subset}
\begin{problem}[Maximal Compatible Subset (MCS)]
We are given a collection of convex polytopes $A = \left\{ \mathcal{P}_1, \mathcal{P}_2, \ldots, \mathcal{P}_n \right\}$ in $\mathbb{R}^k$ such that $0 \in \mathbb{R}^k$ is on the boundary of every $\mathcal{P}_i$, i.e. $0 \in \partial \mathcal{P}_i$ for $i=1, 2, \ldots, n$. The desired output is a maximal subcollection $B = \left\{ \mathcal{P}_{j_1}, \mathcal{P}_{j_2}, \ldots, \mathcal{P}_{j_m} \right\} \subseteq A$ with the following property
\begin{equation}
0 \in \partial\ \mbox{\emph{conv}} \left\{ \bigcup_{\mathcal{P} \in B} \mathcal{P} \right\}.
\end{equation}

\end{problem}

\begin{theorem}\label{thm:nphard}
MCS is NP-hard.
\end{theorem}
\begin{proofsketch}
We use a direct reduction to \textsc{Max 3-Sat}. Let $\Phi(w_1, w_2, \ldots, w_k) = \bigwedge_{i=1}^n \tau_i $ be a formula in the 3-conjunctive normal form with clauses
\begin{equation}
\tau_i = q^1_i\ w_{a_i} \vee q^2_i\ w_{b_i} \vee q^3_i\ w_{c_i},
\end{equation}
where $w_j$ are binary variables, $1 \leq a_i, b_i, c_i \leq k$, and $q^1_i, q^2_i, q^3_i \subseteq \{\neg\}$. For every clause $\tau_i$, we build 8 convex polytopes $\mathcal{P}^{\tt 000}_i, \ldots, \mathcal{P}^{\tt 111}_i$ in $\mathbb{R}^k$, where the superscripts are in binary. To achieve that, we first build a \emph{true} $\mathcal{T}^{\tt 1}_j$ and a \emph{false} $\mathcal{T}^{\tt 0}_j$ convex polytope for every variable $w_j$. Let $\left\{e_1, e_2, \ldots, e_k\right\}$ be the standard orthonormal basis for $\mathbb{R}^k$ and define
\begin{eqnarray}
\mathcal{T}^{\tt 1}_j &=& \mbox{conv}\left\{ 0, M e_j\pm e_1, \ldots, M e_j\pm e_{j-1}, M e_j\pm e_{j+1}, \ldots, M e_j\pm e_k  \right\},\\
\mathcal{T}^{\tt 0}_j &=& \mbox{conv}\left\{ 0, -M e_j\pm e_1, \ldots, -M e_j\pm e_{j-1}, -M e_j\pm e_{j+1}, \ldots, -M e_j\pm e_k  \right\} = -\mathcal{T}^{\tt 1}_j,\\
\end{eqnarray}
for an arbitrary $1 \ll M \in \mathbb{Z}$ and $j=1, 2, \ldots, k$. Note that $\mathcal{T}$ is a $k$-dimensional narrow arrow with polynomial complexity, and its vertices are on the integer lattice and can be computed in polynomial time; see Figure \ref{fig:arrow}. Moreover, $\{\mathcal{T}^{\tt 0}_j, \mathcal{T}^{\tt 1}_j \}$ is incompatible, i.e. $0 \not \in \partial\ \mbox{conv}\{\mathcal{T}^{\tt 0}_j \cup \mathcal{T}^{\tt 1}_j \}$. Let
\begin{equation}
\mathcal{P}^{\tt d_1d_2d_3}_i = \mbox{conv}\left\{ \mathcal{T}^{\tt d_1}_{a_i} \cup \mathcal{T}^{\tt d_2}_{b_i} \cup \mathcal{T}^{\tt d_3}_{c_i} \right\}.
\end{equation}
Note that eventhough convex hull is NP-hard in general \cite{Dye83}, $\mathcal{P}^{\tt d_1d_2d_3}_i$ can be computed in polynomial time. Essentially, the vertex set of $\partial \mathcal{P}^{\tt d_1d_2d_3}_i$ is the union of vertices of $\partial \mathcal{T}^{\tt d_j}_{a_i}$. Assume $({\tt \sigma^1_i, \sigma^2_i, \sigma^3_i}) \in \{{\tt 0, 1}\}^3$ is that assignment to variables $(w_{a_i}, w_{b_i}, w_{c_i})$ which makes $\tau_i$ false, and let
\begin{equation}
A_i = \left\{\mathcal{P}^{\tt 000}_i, \ldots, \mathcal{P}^{\tt 111}_i\right\} \backslash \left\{\mathcal{P}^{\tt \sigma^1_i\sigma^2_i\sigma^3_i}_i\right\}.
\end{equation}
More precisely,
\begin{equation}
{\tt \sigma^j_i} = \left\{ \begin{array}{l}
{\tt 0} \quad \text{if $q^j_i = \emptyset$,} \\
{\tt 1} \quad \text{otherwise.} 
\end{array} \right.
\end{equation}
Note that for every $\mathcal{P}^{\tt d_1d_2d_3}_i \in A_i$, assignment of $\tt (d_1, d_2, d_3)$ to variables $(w_{a_i}, w_{b_i}, w_{c_i})$ makes $\tau_i$ true. Finally, define the input to the MCS as 
\begin{equation}
A = \bigcup_{i=1}^n A_i,
\end{equation}
and assume $B$ is a maximal subset of $A$ with the property
\begin{equation}
0 \in \partial\ \mbox{conv} \left\{ \bigcup_{\mathcal{P} \in B} \mathcal{P} \right\}.
\end{equation}

\begin{figure}[t!]
\begin{center}
\includegraphics[width=4in]{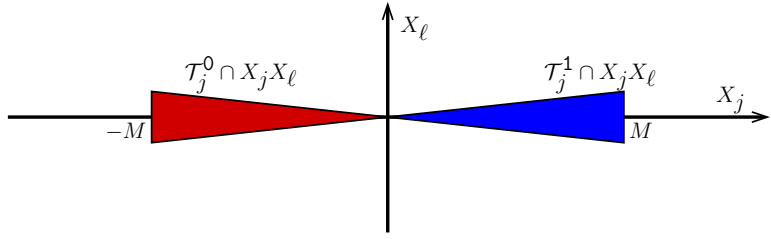}
\end{center}
\caption{Intersection of $\mathcal{T}^{\tt 0}_j$ and $\mathcal{T}^{\tt 1}_j$ polytopes with the $X_jX_\ell$ plane in the proof of Theorem \ref{thm:nphard}. }
\label{fig:arrow}
\end{figure}

It is sufficient to show that $|B| = $ \textsc{Max 3-Sat}($\Phi$). First, we prove that $|B| \geq $ \textsc{Max 3-Sat}($\Phi$). Suppose \textsc{Max 3-Sat}($\Phi$) = $m$, and the assignment $({\tt g_1, g_2, \ldots, g_k}) \in \{{\tt 0, 1}\}^k$ renders $\tau_{i_1}, \tau_{i_2}, \ldots, \tau_{i_m}$ true. Define 
\begin{equation}
B_m = \left\{\mathcal{P}^{{\tt g}_{a_i}{\tt g}_{b_i}{\tt g}_{c_i}}_{i} \ |\ i \in \{i_1, i_2, \ldots, i_m\}  \right\} \subseteq A,
\end{equation}
and verify that it is compatible, i.e. $0 \in \partial\ \mbox{conv} \left\{ \bigcup_{\mathcal{P} \in B_m} \mathcal{P} \right\}$, since
\begin{equation}
\mbox{conv} \left\{ \bigcup_{\mathcal{P} \in B_m} \mathcal{P} \right\} = \mbox{conv} \left\{ \bigcup_{j=1}^k \mathcal{T}^{\tt g_j}_j \right\}
\end{equation}
spans only (approximately) one orthant of $\mathbb{R}^k$. Since $B$ is a maximal subset of $A$ with the compatibility property, $|B| \geq |B_m| = m$. Second, we prove that $|B| \leq $ \textsc{Max 3-Sat}($\Phi$). We show that for every $1 \leq i \leq n$, at most one $\mathcal{P}_i \in B$. To the contrary, suppose for some $1 \leq i \leq n$, $\mathcal{P}^{\tt d_1d_2d_3}_i,  \mathcal{P}^{\tt d'_1d'_2d'_3}_i \in B$. Without loss of generality, assume $d_1 \not = d'_1$. In that case, 
\begin{equation}
\begin{split}
0 & \in \mbox{int}\left\{\mbox{conv}\left[\mathcal{T}^{\tt 0}_{a_i} \cup \mathcal{T}^{\tt 1}_{a_i} \right]\right\} \\ 
& = \mbox{int}\left\{\mbox{conv}\left[\mathcal{T}^{\tt d_1}_{a_i} \cup \mathcal{T}^{\tt d'_1}_{a_i} \right]\right\} \\
& \subseteq \mbox{int}\left\{\mbox{conv}\left[\mathcal{P}^{\tt d_1d_2d_3}_i \cup \mathcal{P}^{\tt d'_1d'_2d'_3}_i \right]\right\} \\
& \subseteq \mbox{int}\left[\mbox{conv}\left\{ \bigcup_{\mathcal{P} \in B} \mathcal{P} \right\}\right],
\end{split}
\end{equation}
which is a contradiction. Above $\mbox{int}$ denotes the interior. Similarly, $B$ induces a consistent assignment to the variables, that makes $|B|$ clauses true. \qed
\end{proofsketch}

\subsection{Randomized greedy algorithm}
Since the maximal compatible subset problem is NP-hard in general, we used a randomized greedy algorithm. In the $i^\text{th}$ iteration, our algorithm starts with a seed $B_i$ which is a random subset of $A$, the input set of polytopes. In our case, $B_i$ is a single element subset. The algorithm iteratively keeps adding other members of $A$ to $B_i$ as long as $0$ remains as a vertex of the convex hull of union of all of the polytopes in $B_i$. Note that Theorem \ref{thm:cond} requires $0$ to be on the boundary of the convex hull, not necessarily on a vertex. However in practice, applicable cases have $0$ as a vertex. Finally, the $B_i$ with maximum number of elements is returned in the output; see Algorithm \ref{alg:rgmcs}.

\begin{algorithm}[t!]
\begin{algorithmic}
\State {\bf Input:} $A = \left\{\mathcal{P}_1, \mathcal{P}_2, \ldots, \mathcal{P}_n \right\}$
\State {\bf Output:} $B \subseteq A$
\State
\State $i \gets$ 1
\State $B \gets \emptyset$
\While {$i < $ \textsc{MaxIterations}} \Comment{\textsc{MaxIterations} is a static/dynamic constant}
	\State $A' \gets \left\{\mathcal{P} \in A\ |\ 0 \mbox{ is a vertex of } \partial \mathcal{P} \right\}$ \Comment{remove inapplicable polytopes}
	\State $B_i \gets \emptyset$ \Comment{empty subcollection}
	\State $\mathcal{Q} \gets \emptyset$ \Comment{empty polytope}
 
	\While{$A' \not = \emptyset$}
		\State $\mathcal{P}_r \gets \textsc{Random}(A')$ \Comment{pick a random polytope from $A'$}
		\State $A' \gets A' \backslash \left\{ \mathcal{P}_r \right\}$
		\State $C \gets B_i\cup\{\mathcal{P}_r \}$
		\State $\mathcal{R} \gets \mbox{conv}\left( \mathcal{Q} \cup \mathcal{P}_r \right)$
		\If {$0$ is a vertex of $\partial \mathcal{R}$}
			\State $B_i \gets C$ \Comment{greedily expand the subcollection}
			\State $\mathcal{Q} \gets \mathcal{R}$ \Comment{update the convex hull of the union}
		\EndIf
	\EndWhile
	\If {$|B| < |B_i| $}
		\State $B \gets B_i$
	\EndIf
	\State $i \gets i+1$
\EndWhile 
\end{algorithmic}
\caption{Randomized Greedy Maximal Compatible Subset}\label{alg:rgmcs}
\end{algorithm}

\begin{figure}[t!]
\begin{center}
\vspace*{-1.0in}
\includegraphics[width=4in]{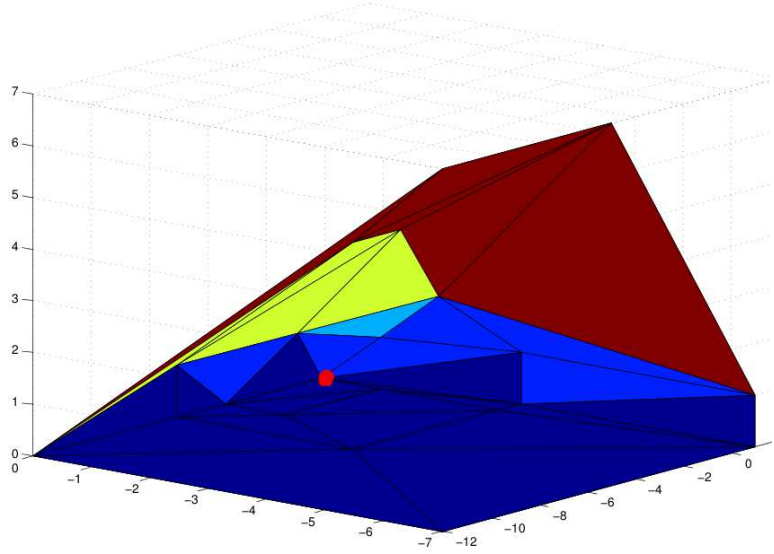}
\vspace*{-1.0in}
\end{center}
\caption{$\mathcal{N}(D')$ where $D'$ consists of three sequence-structure pairs: PDB\_00434, PDB\_00200, and PDB\_00876. The dot shows the origin $0$. Secondary structures in $D'$ are shown in Figure \ref{fig:structures}.}
\label{fig:polytope}
\end{figure}

\begin{figure}[t!]
\begin{center}
\begin{tabular}{|ccc|}
\hline
&&\\
\hspace{-1.1in}\parbox{2in}{\vspace{-0.5in}\includegraphics[scale=0.25]{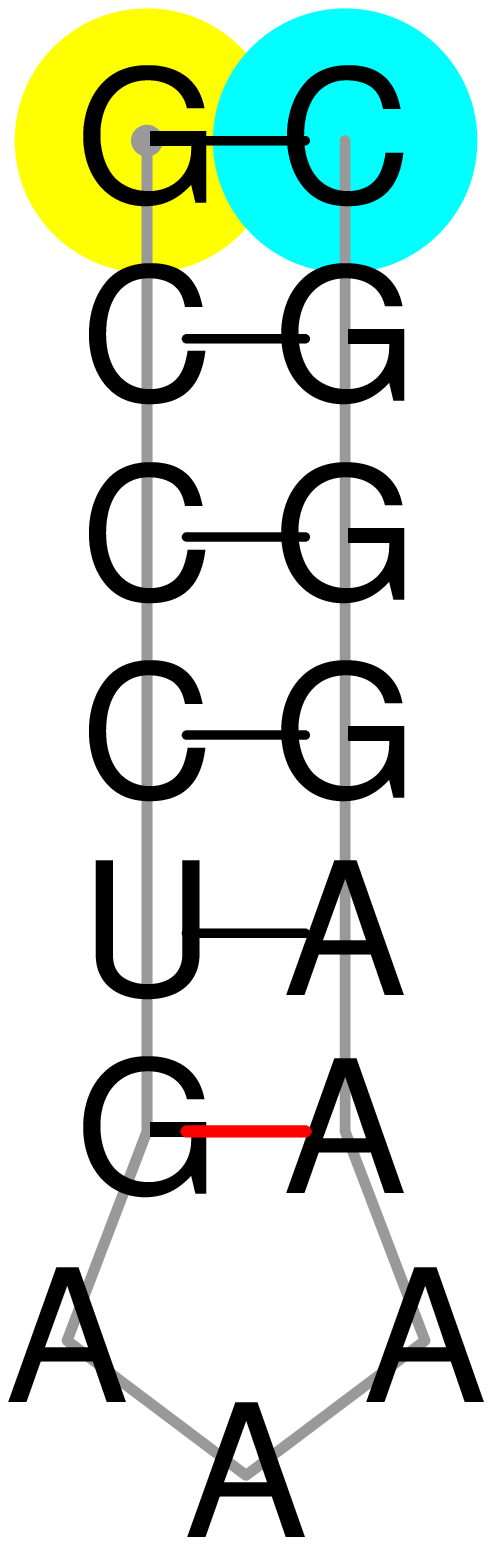}\vspace{-0.6in}} &
\hspace{-1.5in}\parbox{2in}{\vspace{-0.5in}\includegraphics[scale=0.35]{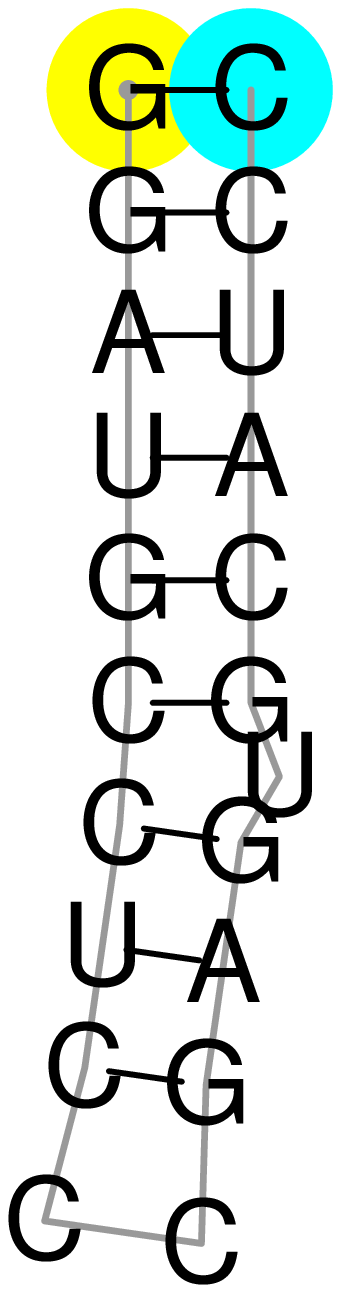}\vspace{-1.1in}}&
\hspace{-1.3in}\parbox{2in}{\vspace{-0.5in}\includegraphics[scale=0.35]{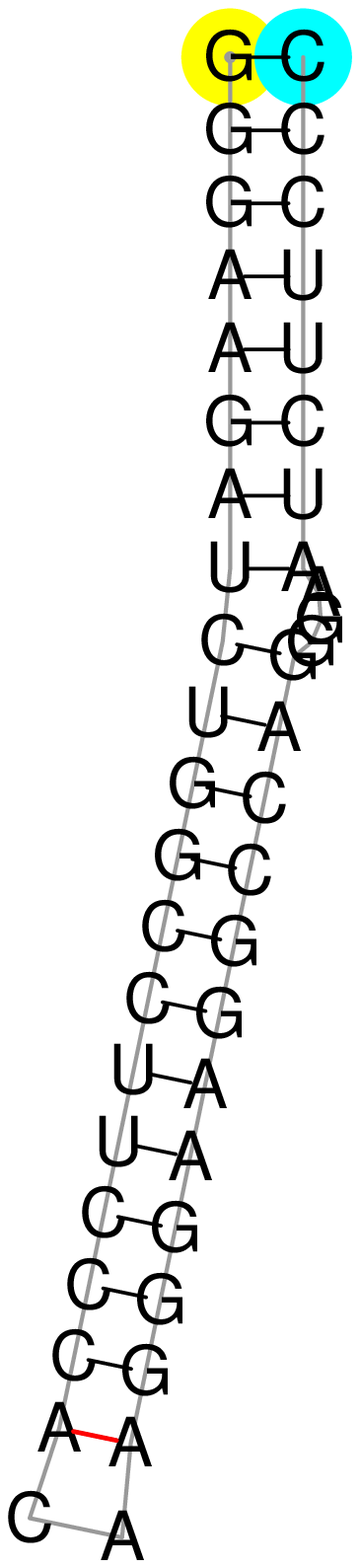}\vspace{-1.0in}} \\
PDB\_00434 & PDB\_00200 & PDB\_00876 \\
\hline
\end{tabular}
\end{center}
\caption{Structures of PDB\_00434, PDB\_00200, and PDB\_00876 obtained from RNA STRAND v2.0 website \cite{Andronescu08}. }
\label{fig:structures}
\end{figure}

\section{Results}
We used 2277 unpseudoknotted RNA sequence-structure pairs from RNA STRAND v2.0 database as our data set $D$. RNA STRAND v2.0 is a convenient source of RNA sequences and structures selected from various Rfam families \cite{Burge13} and contains known RNA secondary structures of any type and organism, particularly with and without pseudoknots \cite{Andronescu08}. There are 2334 pseudoknot-free RNAs in the RNA STRAND database. We sorted them based on their length and selected the first  2277 ones for computational convenience. We excluded pseudoknotted structures because our current implementation is incapable of considering pseudoknots. Some sequences in the data set allow only A-U base pairs (not a single C-G or G-U pair), in which case the Newton polytope degenerates into a line.

We demonstrate the results for the separate A-U, C-G, and G-U base pair counting energy model similar to our previous model \cite{Forouzmand13}. In that case, the feature vector $$c(x, s)=(c_1(x, s), c_2(x, s), c_3(x, s))$$ is three dimensional: $c_1(x, s)$ is the number of A-U, $c_2(x, s)$ the number of C-G, and $c_3(x, s)$ the number of G-U base pairs in $s$. First, we computed $c(x, y)$ and used our Newton polytope program to compute $\mathcal{N}(x)$ for each $(x, y) \in D$ \cite{Forouzmand13}. For completeness, we briefly include our dynamic programming algorithm which starts by computing the Newton polytope for all subsequences of unit length, followed by all subsequences of length two and more up to the Newton polytope for the entire sequence $x$. We denote the Newton polytope of the subsequence $n_i\cdots n_j$ by $\mathcal{N}(i, j)$, i.e.
\begin{equation}
\mathcal{N}(i, j) := \mathcal{N}(n_i\cdots n_j).
\end{equation} 
The following dynamic programming yielded the result
\begin{equation}\label{equ:dp}
\mathcal{N}(i, j) = \mbox{conv}\left\{\bigcup  
\left[ \begin{array}{l}
    \mathcal{N}(i,\ell) \oplus \mathcal{N}(\ell+1,j),\; \ \ i \leq \ell \leq j-1\\
    \left\{(1, 0, 0) \right\} \oplus \mathcal{N}(i+1,j-1) \quad \text{if $n_in_j$ = AU$\ |\ $UA}\\
    \left\{(0, 1, 0) \right\} \oplus \mathcal{N}(i+1,j-1) \quad \text{if $n_in_j$ = CG$\ |\ $GC}\\
    \left\{(0, 0, 1) \right\} \oplus \mathcal{N}(i+1,j-1) \quad \text{if $n_in_j$ = GU$\ |\ $UG}\\
    \end{array} \right] \right\},
\end{equation}
with the base case $\mathcal{N}(i, i) = \left\{ (0, 0, 0) \right\}$. Above $\oplus$ is the Minkowski sum. We then computed $\mathcal{N}_y(x)$ by translating the Newton polytope so that $c(x, y)$ moves to the origin $0$, to obtain $A = \{\mathcal{N}_y(x)\ |\ (x, y) \in D \}$ as the input to Algorithm \ref{alg:rgmcs}. 

We then removed those polytopes in $A$ that do not have $0$ as one of their boundary vertices, to obtain $A'$ in Algorithm \ref{alg:rgmcs}. It turns out that only 126 sequences out of the initial 2277 remain in $A'$. Note that our condition here is more stringent than the necessary condition in \cite{Forouzmand13}, and that is why fewer sequences satisfy this condition. After 100 iterations (\textsc{MaxIterations} = 100) which took less than a minute, the algorithm returned 3 polytopes that are compatible, i.e. the origin $0$ is a vertex of the boundary of convex hull of union of three polytopes. They correspond to sequences PDB\_00434 (length: 15nt; bacteriophage HK022 nun-protein-nutboxb-RNA complex \cite{Faber01}), PDB\_00200 (length: 21nt; an RNA hairpin derived from the mouse 5' ETS that binds to the two N-terminal RNA binding domains of nucleolin \cite{Finger03}), and PDB\_00876 (length: 45nt; solution structure of the HIV-1 frameshift inducing element \cite{Staple05}) which are experimentally verified 
by NMR or X-ray. 

The origin is incident to 5 facets of the resulting convex hull, shown in Figure \ref{fig:polytope}, the inward normal vector to which are in the rows of
\begin{equation}
J =
\begin{bmatrix}
-0.3162 & -0.9487 & 0 \\
-1  & 0 & 0 \\
0 & 0 & 1 \\
-0.4082 & -0.8165 & -0.4082\\
-0.5774 & -0.5774 & -0.5774\\
\end{bmatrix}
\end{equation}
as explained in Theorem \ref{thm:cone}. The set of those energy parameters that correctly predict the three structures in Figure \ref{fig:structures} is the convex cone generated by these 5 vectors. The Turner model measures the average energies of A-U, C-G, and G-U to be approximately $(-2, -3, -1)$ kcal/mol \cite{MatTur99}. To test whether those energy parameters fall into the convex cone generated by $J$, we solved a convex linear equation. Let \[ \mathbf{h}^\dagger= \begin{bmatrix}-2 \\ -3 \\ -1 \end{bmatrix}. \] We would like to find a positive solution $v \in \mathbb{R}_+^5$ to the linear equation $J^T v = \mathbf{h}^\dagger$. We formulated that as a linear program which
was solved using GNU Octave, and here is the answer: 
\[
v =
\left[ {\begin{array}{cc}
2.10815  \\
0.33340\\
0\\
0\\
1.73190\\
\end{array} } \right]. \]
The convex cone generated by $J$ contains the vector $(-2, -3, -1)$. Therefore, our finding is in agreement with the Turner base pairing energies. Note that the three structures are base-pair rich (Figure \ref{fig:structures}).

\section{Discussion}
We further developed the notion of learnability of parameters of an energy model. A necessary and sufficient condition for it was given, and a characterization of the set of energy parameters that realize exact structure prediction followed as a by-product. If an energy model satisfies the sufficient condition, then we say that the training set is compatible. In our case, the RNA STRAND v2.0 training set is not compatible (for the A-U, C-G, G-U base pair counting model). We showed that computing a maximal compatible subset of a set of convex polytopes is NP-hard in general and gave a randomized greedy algorithm for it. The computed set of energy parameters for A-U, C-G, G-U from a maximal compatible subset agreed with  the thermodynamic energies. Complexity of the MCS problem is an open and interesting question, particularly if we treat the dimension of the feature space as a constant. Also, assessing the generalization power of an energy model remains for future work.

\bibliographystyle{unsrt}
\bibliography{ref,pub,main}

\end{document}